\documentclass[conference]{IEEEtran}
\IEEEoverridecommandlockouts
\usepackage[noadjust]{cite}

\usepackage{amsmath,amssymb,amsfonts}
\usepackage{algorithmic}
\usepackage{graphicx}
\usepackage{textcomp}
\usepackage{xcolor}
\usepackage{colortbl}
\def\BibTeX{{\rm B\kern-.05em{\sc i\kern-.025em b}\kern-.08em
    T\kern-.1667em\lower.7ex\hbox{E}\kern-.125emX}}
\usepackage{tikz}
\usetikzlibrary{arrows}
\usetikzlibrary{calc}
\usepackage{amsthm}
\usepackage{thmtools}
\declaretheorem{theorem}
\declaretheorem[name=Lemma, style=theorem]{lemma}
\declaretheorem[name=Proposition, style=theorem]{proposition}
\declaretheorem[name=Corollary, style=theorem]{corollary}
\declaretheorem[name=Definition, style=definition, qed=$\blacksquare$]{definition}
\declaretheorem[name=Example, style=definition, qed=$\blacksquare$]{example}
\declaretheorem[name=Conjecture, style=theorem]{conjecture}
\usepackage{mathpartir}
\usepackage{todonotes}
\usepackage{caption}
\begin{document}

\newcommand{\Nat}{\mathbb{N}}
\newcommand{\ME}{\text{ME}}
\newcommand{\MEF}{\text{MEF}}
\newcommand{\project}{\operatorname{\downarrow}}
\renewcommand{\L}{\mathcal{L}}
\newcommand{\lequiv}[1]{\approx_{#1}}
\newcommand{\Pow}[1]{\mathcal P(#1)}
\newcommand{\Bool}{\texttt{Bool}}
\newcommand{\rel}[1]{\llbracket #1 \rrbracket}
\newcommand{\sem}[1]{{#1}}
\newcommand{\many}[1]{\overline{#1}}
\renewcommand{\iff}{\Leftrightarrow}
\renewcommand{\implies}{\Rightarrow}
\newcommand{\subst}[3]{[{#2}/{#1}]{#3}}
\newcommand{\closed}[1]{\texttt{closed}({#1})}
\newcommand{\return}{\texttt{return}}
\newcommand{\Sec}{\texttt{Sec}}
\newcommand{\DCC}{\texttt{DCC}}
\newcommand{\up}{\texttt{up}}
\newcommand{\True}{\texttt{true}}
\newcommand{\False}{\texttt{false}}
\newcommand{\If}{\texttt{if}}
\newcommand{\bind}{\texttt{bind}}
\newcommand{\pair}{\texttt{pair}}
\newcommand{\either}{\texttt{either}}
\newcommand{\case}{\texttt{case}}
\newcommand{\canFlowTo}{\sqsubseteq}
\renewcommand{\S}{S}
\newcommand{\D}{D}
\newcommand{\below}{\preceq}
\newcommand{\safe}{\texttt{safe}}
\newcommand{\map}{\texttt{map}}
\newcommand{\dyn}{\texttt{toDynamic}}
\newcommand{\stat}{\texttt{toStatic}}
\newcommand{\red}{\longrightarrow}
\newcommand{\lub}{\sqcup}
\newcommand{\exFalso}{\texttt{ex-falso}}
\newcommand{\cs}{{\rel{S}}}
\newcommand{\CC}{\texttt{CC}}
\newcommand{\Obj}{\texttt{F}_{\omega o}}
\newcommand{\ObjBase}{\texttt{F}_{\omega}}
\newcommand{\elim}{\texttt{-elim}}
\newcommand{\prot}{\texttt{protect}}
\newcommand{\public}{\texttt{public}}
\newcommand{\unit}{\texttt{<>}}
\newcommand{\trans}{\texttt{trans}}
\newcommand{\join}{\texttt{join}}
\newcommand{\of}{\texttt{of}}
\newcommand{\translates}{\hookrightarrow}
\newcommand{\com}{\texttt{com}}
\newcommand{\defined}[2]{#1(#2)^\checkmark}
\newcommand{\notdefined}[2]{#1(#2)^\times}
\newcommand{\refinedBy}{\sqsubseteq}
\newcommand{\TX}{\ensuremath{\tau}}
\newcommand{\pto}{\to_p}
\newcommand{\inl}[1]{\texttt{inl}(#1)}
\newcommand{\inr}[1]{\texttt{inr}(#1)}
\newcommand{\MT}{\text{MT}}
\newcommand{\coMT}{\text{coMT}}
\newcommand{\TI}{\text{TI}}
\newcommand{\TS}{\text{TS}}
\newcommand{\Total}{\text{Total}}
\newcommand{\METI}{\text{MEST}}
\newcommand{\udef}{\text{diverge}}
\newcommand{\Alice}{\text{Alice}}
\newcommand{\Bob}{\text{Bob}}
\newcommand{\Charlie}{\text{Charlie}}
\newcommand{\LOW}{\ensuremath{{L}}}
\newcommand{\MED}{\ensuremath{{M}}}
\newcommand{\HIGH}{\ensuremath{{H}}}
\newcommand{\inputify}{\text{inputify}}
\newcommand{\outputify}{\text{outputify}}
\newcommand{\nameify}[1]{\textit{#1}}
\newcommand{\divIfPre}{\nameify{divergeIfHPresent}}
\newcommand{\divIfAbs}{\nameify{divergeIfHAbsent}}
\newcommand{\id}{\nameify{id}}
\newcommand{\leakAll}{\nameify{leakAll}}
\newcommand{\leakBit}{\nameify{leakBit}}
\newcommand{\terminationLeak}{\nameify{termLeak}}
\newcommand{\divOnLow}{\nameify{divergeIfLPresent}}
\newcommand{\salaryAverages}{\nameify{salaryAverages}}
\newcommand{\combine}{\nameify{combine}}
\newcommand{\leakLevel}{\nameify{leakLevel}}
\newcommand{\combineAll}{\nameify{combineAll}}
\newcommand{\vlinesep}[2]{\smash{\vrule width .5pt depth #1 height #2}}
\newcommand{\defAs}{\ensuremath{\triangleq}}
\newcommand{\LA}{{\mathbb{L}}}
\newcommand{\divIfEmpty}{\nameify{divergeIfEmpty}}
\newcommand{\nums}{\#}
\newcommand{\toInput}{\nameify{toInput}}

\title{Transparent IFC Enforcement:\\Possibility and (In)Efficiency Results
\thanks{The first author was partially supported by the Wallenberg AI,
Autonomous Systems and Software Program (WASP) funded by the Knut and Alice
Wallenberg Foundation and The Osher Endowment's through project Ghost:
Exploring the Limits of Invisible Security. The second author was partially
supported by NSF grant CCF-1813133
}
}

\author{\IEEEauthorblockN{Maximilian Algehed}
\IEEEauthorblockA{\textit{Computer Science and Engineering} \\
\textit{Chalmers}\\
G\"oteborg, Sweden \\
\texttt{algehed@chalmers.se}}
\and
\IEEEauthorblockN{Cormac Flanagan}
\IEEEauthorblockA{\textit{Computer Science and Engineering} \\
\textit{University of California Santa Cruz}\\
Santa Cruz, USA \\
\texttt{cormac@ucsc.edu}}
}

\maketitle

\begin{abstract}
  Information Flow Control (IFC) is a collection of techniques for ensuring
a no-write-down no-read-up style security policy known as \emph{noninterference}.
Traditional methods for both static (e.g. type systems) and dynamic (e.g.
runtime monitors) IFC suffer from untenable numbers of false alarms on
real-world programs.
Secure Multi-Execution (SME) promises to provide secure information flow
control without modifying the behaviour of already secure programs, a property
commonly referred to as \emph{transparency}.
Implementations of SME exist for the web in the form of the FlowFox browser and
as plug-ins to several programming languages.
Furthermore, SME can in theory work in a black-box manner, meaning that it can
be programming language agnostic, making it perfect for securing legacy or
third-party systems.
As such SME, and its variants like Multiple Facets (MF) and Faceted Secure
Multi-Execution (FSME), appear to be a family of panaceas for the security engineer.
The question is, how come, given all these advantages, that these techniques
are not ubiquitous in practice?

The answer lies, partially, in the issue of runtime and memory overhead.
SME and its variants are prohibitively expensive to deploy in many non-trivial
situations.
The natural question is \emph{why} is this the case?
On the surface, the reason is simple.
The techniques in the SME family all rely on the idea of
\emph{multi-execution}, running all or parts of a program multiple times to
achieve noninterference.
Naturally, this causes some overhead.
However, the predominant thinking in the IFC community has been that these
overheads can be overcome.
In this paper we argue that there are fundamental reasons to expect this
not to be the case and prove two key theorems:
\begin{itemize}
  \item All transparent enforcement is polynomial time equivalent to multi-execution.
  \item All black-box enforcement takes time exponential in the number of principals
    in the security lattice.
\end{itemize}
Our methods also allow us to answer, in the affirmative, an open question about
the possibility of secure and transparent enforcement of a security condition
known as Termination Insensitive Noninterference.

\end{abstract}


\section{\label{sec:introduction} Introduction}
Language-Based Information Flow Control (IFC) \cite{sabelfeld2003language,
Denning,FlowCaml,Jif} is a promising technology for securing systems
against malicious third-party code.
Approaches to IFC typically appear in the form of either a type system
\cite{Jif, DCC, FlowCaml} or a custom programming language semantics \cite{LIO,
MF}.
Traditionally, these approaches either statically or dynamically detect
behaviour that violates the security criteria of \emph{Noninterference}
\cite{goguen1982security,mclean1992proving,zdancewic2003observational,sabelfeld2003language}
(secrets cannot influence attacker-observable
behaviour) and either raise a type or runtime error.
However, in order to be sound, even in the state-of-the-art implementations,
these techniques need to be conservative and therefore suffer from unmanageable
numbers of false alarms \cite{King, staicu2019empirical}.

To remedy this issue, techniques have recently been developed for ensuring
so called \emph{transparent} IFC
\cite{SME, MF, FSME, OGMF, OptimisingFSME, zanarini2013precise}.
This line of work promises to provide security without raising false alarms.
Because precisely detecting violations of noninterference is impossible
\cite{Denning} to do both statically and dynamically, methods for transparent
IFC instead silently modify programs to ensure noninterference by construction.

All known techniques for providing transparent IFC share one feature, they
are all based on multi-execution.
This technique was pioneered by Devriese and Piessens \cite{SME} in their seminal
paper on Secure Multi-Execution (SME).
Under SME, the program being made secure is run once for each security level 
with carefully adapted inputs to ensure noninterference while removing false
alarms.

To see an example of how this works, consider the two security levels $\LOW$
(for Low or Public) and $\HIGH$ (for High or Secret).
Noninterference stipulates that low input is allowed to influence high output,
but not the other way around.
Under SME, a program $p$ that takes both high and low inputs and produces both
high and low output is run twice; one run of $p$ is given only the low input
and produces the low part of the output, and the other run of $p$ is given both
low and high input and contributes only the high output: see Figure \ref{fig:ME}.

SME is a black-box enforcement mechanism, it does not need access to the
source code of $p$ to work.
The ability to secure any program in a black-box manner is powerful and has
potential applications in many areas, including databases
\cite{yang2016precise}, legacy code \cite{pfeffer2019efficient}, and browsers
\cite{de2012flowfox}.

The black-box property of SME is shared by some \cite{rafnsson2016secure,
zanarini2013precise}, but not all \cite{OGMF, FSME, MF} transparent enforcement
mechanisms in the literature.
However, the idea of multi-execution, running the same code multiple times with
slightly different inputs, is shared among all known mechanisms in one form or
another.
Consequently, these mechanisms, no matter where they lie on the scale from
black-box to white-box, suffer severe performance penalties as the number of
multi-executions grows \cite{MF, FSME, OptimisingFSME}.

In this paper, we provide a unifying formal, extensional, framework for
studying multi-execution that is sufficiently expressive to formulate and prove
a number of theorems that explain:
\begin{enumerate}
  \item Why all transparent enforcement mechanisms rely on multi-execution,
  \item Why all transparent black-box enforcement mechanisms are inefficient
    (as seen empirically in previous work, e.g. \cite{MF, FSME, OptimisingFSME, wong2018faster})
\end{enumerate}

Along the way, we use our framework to answer the open question of whether
or not Termination Insensitive Noninterference can be transparently enforced.
It turns out that it is possible, but inefficient.

Concretely, we provide the following contributions:
\begin{itemize}
  \item A novel yet simple framework for reasoning about secure programs
    and enforcement mechanisms (Section \ref{sec:def-enforcement-mechanisms}).
  \item A precise characterisation of the transparency guarantees provided
    by multi-execution (Section \ref{sec:multi-execution}).
  \item A black-box version of secure multi-execution for infininte
    lattices (Section \ref{sec:multi-execution}).
  \item A novel enforcement mechanism that is sound and transparent
    for Termination Insensitive Noninterference (Section \ref{sec:TITI}).
  \item A general framework for optimising multi-execution (Section
    \ref{sec:level-assignments}) in which we prove a number of results
    including conditions for safe and transparent optimisation.
  \item A proof that any secure and transparent enforcement is poly-time
    equivalent to a multi-execution based enforcement mechanism (Section
    \ref{sec:level-assignments}).
  \item A proof that all secure and transparent black-box enforcement mechanisms
    are inefficient (Section \ref{sec:efficiency}).
\end{itemize}

\begin{figure}[t]
  \centering
\begin{tikzpicture}
  \draw (-3, -2) rectangle (3, 2);      

  \draw [line width = 0.25mm, red,  ->] (-4,  1.25) -- node [above] {\HIGH} (-3,  1.25); 
  \draw [line width = 0.25mm, red,  ->] (3,   1.25) -- node [above] {\HIGH} (4,   1.25); 

  \draw [line width = 0.25mm, blue, ->] (-4, -1.25) -- node [below] {\LOW}  (-3, -1.25); 
  \draw [line width = 0.25mm, blue, ->] (3,  -1.25) -- node [below] {\LOW}  (4,  -1.25); 

  \draw (-0.75,  0.5) rectangle (0.75,  1.5) node[pos=0.5]{$p$}; 
  \draw (-0.75, -0.5) rectangle (0.75, -1.5) node[pos=0.5]{$p$}; 

  \draw [line width = 0.25mm, red,  ->] (-3,  1.25) -- (-0.75,  1.25); 
  \draw [line width = 0.25mm, blue, ->] (-3, -1.25) -- (-0.75, -1.25); 

  \draw [line width = 0.25mm, blue, ->] (-3, -1.25) .. controls (-0.75, -1.25) and (-3,  0.75) .. (-0.75, 0.75); 
  \draw [line width = 0.25mm, black!45, ->] (-1.5, -0.75) -- (-0.75, -0.75);

  \draw [line width = 0.25mm, red,  ->] (0.75,  1.25) -- (3,  1.25); 
  \draw [line width = 0.25mm, blue, ->] (0.75, -1.25) -- (3, -1.25); 

  \draw [line width = 0.25mm, black!45, >=angle 90, ->] (0.75, -0.75) -- (1.40, -0.75); 
  \draw [line width = 0.25mm, black!45, >=angle 90, -<] (0.75, -0.75) -- (1.5, -0.75); 
  \draw [line width = 0.25mm, black!45, >=angle 90, ->] (0.75, 0.75)  -- (1.40,  0.75); 
  \draw [line width = 0.25mm, black!45, >=angle 90, -<] (0.75, 0.75)  -- (1.5,  0.75); 
\end{tikzpicture}
  \caption{\normalsize \label{fig:ME} Multi-Execution of the program $p$ for the two-point lattice.}
\end{figure}

\section{\label{sec:framework} An Extensional Framework for Secure Information Flow}
In this section we develop an extensional framework for reasoning about secure
information flow.
The goal is to create a simple, yet flexible, context in which we can reason both
about \emph{possibility} of secure and transparent enforcement as well as
\emph{efficiency}.

\subsection*{\textbf{Programs with Labeled Inputs and Outputs}}

As is usual in IFC research, we assume a join semi-lattice $\langle \L,
\canFlowTo, \bot, \lub \rangle$ that encodes security labels in $\L$ (we use
the words label and level interchangeably), the permitted flows between them by
the order $\_\canFlowTo\_ \subseteq \L \times \L$, the least privileged label
with $\bot \in \L$, and a least-upper-bound operation $\_\lub\_ : \L \times \L \to \L$.
We let the variables $\ell$, $\ell_i$, $\ell'$, $\jmath$, $\jmath_i$, $\jmath'$
etc. range over elements of $\L$ and we write $a^\ell$ for the pair $(a, \ell)$.

We work in a non-interactive setting, so a program takes input and eventually
produces output (or diverges).
The input (and output) is typically a compound data-structure consisting of
various ``pieces'', where each piece can be annotated with its own security label.
For example, the input could be a set of files, where each file has a name, contents,
and a security label.
To model this kind of setting in a general way we assume that both the
input and the output is a set of labeled data, and so we formalise the semantics
of programs as partial recursive functions $p : \Pow{A \times \L} \pto \Pow{B \times \L}$.

Here, $\Pow{\bullet}$ denotes the power-set constructor and $X \pto Y$ denotes
\emph{partial} functions from $X$ to $Y$, while $X \to Y$ denotes total
functions.
One benefit of working with this ``sets of labeled data'' formalism is that
it is simple to remove high-security information from an input or
output $x \in \Pow{A \times \L}$ via the following so-called \emph{$\ell$-projection}
operation:
\begin{align*}
  x \project \ell &= \{\ a^{\jmath}\ |\ a^{\jmath} \in x, \jmath \canFlowTo \ell \}\\
  \intertext{Similarly, we can collect the \emph{labels of $x$} as:}
  \L(x) &= \{\ \ell\ |\ a^\ell \in x \}\\
  \intertext{And finally, we define the \emph{$\ell$-selection of $x$}, i.e.
  all the elements of $x$ at level $\ell$, as:}
  x@\ell &= \{\ a^\ell\ |\ a^\jmath \in x,\ \jmath = \ell \}
\end{align*}

We refer to the \emph{syntax} of $p$ simply as $p$
and the \emph{semantics} of $p$ as $p(x)$ for some input $x$.
Because the choice of programming language is orthogonal to our purposes, we
keep it abstract.
We write $\defined{p}{x}$ when $p(x)$ is defined and write $\notdefined{p}{x}$ to mean
that $p(x)$ is undefined (i.e. that $p$ diverges on input $x$).

Next we formalise what we mean by noninterfering programs.
Intuitively, a program $p$ is noninterfering if, given inputs $x$ and $y$ that the attacker
considers observably equivalent, $p$ produces outputs $p(x)$ and $p(y)$ that the
attacker also considers observably equivalent.
\begin{definition}
  We call $x, y \in \Pow{A \times \L}$ \emph{$\ell$-equivalent}, written $x \sim_\ell y$, if and
  only if the $\ell$-projection of $x$ and $y$ are the same.
  $$
  x \sim_\ell y\ \defAs\ x \project \ell = y \project \ell
  $$
  A program $p$ is \emph{noninterfering} if, given any $\ell$ and two $x$ and $y$
  that are $\ell$-equivalent such that both $p(x)$ and $p(y)$ are defined, we have that
  $p(x)$ is $\ell$-equivalent to $p(y)$.
  $$
  p\ \text{is noninterfering} \defAs
  $$
  $$
  \forall \ell.\forall x, y.\ 
  x \sim_\ell y \wedge \defined{p}{x} \wedge \defined{p}{y} \implies p(x) \sim_\ell p(y)
  $$
\end{definition}

\begin{example}
  \label{ex:running}
  We present a number of example programs that serve as running
  examples throughout the paper.
  First is the program $\id$, the identity function, which is noninterfering.
  \begin{align*}
    \id(x) &\defAs x\\
    \intertext{The next program combines information from multiple
    security labels, but is still noninterfering since the label on
    the output appropriately reflects the dependencies on the labeled
    input. We use $|x|$ to denote the size of the set $x$.}
    \combine(x) &\defAs \{ (|x@\Alice| + |x@\Bob|)^{\Alice \lub \Bob \lub \Charlie} \}\\
    \intertext{Labels $\Alice$, $\Bob$, and $\Charlie$ are notational shorthand
    for labels $\{\Alice\}$, $\{\Bob\}$ and $\{\Charlie\}$ in the power-set lattice $\Pow{P}$
    for some set of \emph{principals} $P$ which contains $\Alice$, $\Bob$, and $\Charlie$.
    \linebreak
    \indent The program \combineAll~below is similar to \combine, but is interfering because
    it combines everything, rather than a specific subset of the inputs.}
    \combineAll(x) &\defAs \{|x|^{\bigsqcup\L(x)}\}\\
    \intertext{To see that \combineAll~is interfering, consider that $\emptyset \sim_\LOW \{0^\HIGH\}$ but:}
    \combineAll(\emptyset) &= \{0^\LOW\} \not\sim_\LOW \{1^\HIGH\} = \combineAll(\{0^\HIGH\})\\
    \intertext{Another example of an interfering program is $\leakBit$, which
    leaks one bit in the $\HIGH$ input to $\LOW$.}
    \leakBit(x) &\defAs \text{if}\ 1^\HIGH \in x\ \text{then}\ \{1^\LOW\}\ \text{else}\ \{0^\LOW\}\\
    \intertext{The program $\leakAll$ below takes $x$ and re-labels
    everything in the input to $\bot$; it is also interfering.}
    \leakAll(x) &\defAs \{\ a^\bot\ |\ a^\ell \in x\ \}\\
    \intertext{Finally, the program $\terminationLeak$ below leaks information via termination,
    but is still noninterfering because our definition of noninterference does not consider the
    termination channel.}
    \terminationLeak(x) &\defAs\ \text{if}\ 1^\HIGH \in x\ \text{then}\ \emptyset\ \text{else}\ \udef
  \end{align*}
  To summarise, programs $\id$, $\combine$ and $\terminationLeak$ are
  noninterfering, while $\combineAll$, $\leakBit$ and $\leakAll$ are not.
\end{example}


\subsection*{\label{sec:def-secure-programs} \textbf{Secure Programs and Termination Criteria}}

The notion of noninterference introduced above does not prevent leaks through
the termination channel, as illustrated by the $\terminationLeak$ example,
and so captures what is normally called ``Termination Insensitive''
noninterference \cite{hedin2012perspective}.
The prevention of termination leaks is a challenging topic for IFC enforcement, so
to study it we introduce the following four \emph{termination criteria}.
\begin{definition}[Termination Criteria]
  \label{def:termination}
  ~
  \begin{itemize}
    \item All programs are \emph{Termination Insensitive} ($\TI$).
    \item Program $p$ is \emph{Monotonically Terminating} ($\MT$) if and
      only if whenever $p(x)$ is defined, so is $p(x\project\ell)$ for all $\ell$.
      $$
      p\ \text{is}\ \MT\ \defAs\ \forall x.\ \defined{p}{x}\ \implies \forall \ell. \defined{p}{x\project\ell}
      $$
    \item Program $p$ is \emph{Termination Sensitive} ($\TS$)
      if and only if for all $\ell$ and $\ell$-equivalent $x$ and $y$,
      $p(x)$ is defined if and only if $p(y)$ is.
      $$
      p\ \text{is}\ \TS\ \defAs\ \forall \ell. \forall x, y. x \sim_\ell y \implies (\defined{p}{x} \iff \defined{p}{y})
      $$
    \item $p$ is $\Total$ if it always terminates.
      $$
      p\ \text{is}\ \Total\ \defAs\ \forall x. \defined{p}{x}
      $$
  \end{itemize}
  
\end{definition}
We let the meta-variable \TX~range over termination criteria.
$$\TX ::= \TI\ |\ \MT\ |\ \TS\ |\ \Total$$

The termination criteria \TS~and \TI~are standard from the IFC literature
\cite{hedin2012perspective, bielova2016spot, ngo2018impossibility}, and
the notion of total programs is also a natural termination criteria.
Our new criteria, \MT, is motivated by the termination requirements of SME,
as we discuss below.
\MT~is a weaker requirement than \TS, but stronger than \TI, which results in the
following ordering of termination criteria.
\begin{proposition}
  \label{prop:termination-ordering}
  The following chain of implications is strict.
  \begin{align*}
              &p\ \text{is \Total}\\
    \implies\ &p\ \text{is \TS}\\
    \implies\ &p\ \text{is \MT}\\
    \implies\ &p\ \text{is \TI}
  \end{align*}
\end{proposition}
Recall that an implication $P \implies Q$ is \emph{strict} if it is
\emph{not} the case that $P \iff Q$.
Using the termination criteria, we define a family of notions of secure program
in a way that separates termination and noninterference.
\begin{definition}[\TX-secure]
  A program is \emph{\TX-secure} if and only if it is $\TX$ and noninterfering.
\end{definition}
The \TX-secure criteria are naturally ordered the same way as the termination criteria.
\begin{proposition}
  \label{prop:security-ordering}
  The following chain of implications is strict.
  \begin{align*}
              &p\ \text{is \Total-secure}\\
    \implies\ &p\ \text{is \TS-secure}\\
    \implies\ &p\ \text{is \MT-secure}\\
    \implies\ &p\ \text{is \TI-secure}
  \end{align*}
\end{proposition}
We re-visit some programs from Example \ref{ex:running} to illustrate
the various termination criteria.
\begin{example}
  \label{ex:running-revisit}
  The programs $\id$ and $\combine$ from Example \ref{ex:running} are
  noninterfering and \Total; hence they are both \Total-secure.
  The program $\terminationLeak$ is noninterfering, but its termination is
  influenced by $\HIGH$ information, and so it is only \TI-secure.
  Program $\leakBit$, $\combineAll$, and $\leakAll$ meanwhile
  are interfering and hence are not \TX-secure for any \TX.
\end{example}
Next we present three noninterfering programs that illustrate
different termination criteria.
\begin{example}
  \label{ex:running-termination}
  First, the following program is not \Total, as it sometimes diverges,
  but divergence only depends on public (\LOW-labeled) information, so it
  is \TS.
  \begin{align*}
    \divOnLow(x) &\defAs \text{if}\ 1^\LOW \in x\ \text{then}\ x\ \text{else}\ \udef\\
  \intertext{
    In the next program, divergence depends on $\HIGH$ information, so it is not \TS.
    It is \MT~as removing $\HIGH$ information only improves termination.
    }
    \divIfPre(x) &\defAs \text{if}\ 1^\HIGH \in x\ \text{then}\ \udef\ \text{else}\ \emptyset\\
  \intertext{
    In contrast, for the following program, removing $\HIGH$ information hurts termination, so
    this program is \TI~but not \TS.
  }
    \divIfAbs(x) &\defAs \text{if}\ 1^\HIGH \not\in x\ \text{then}\ \udef\ \text{else}\ \emptyset
  \end{align*}
  For a brief summary of this example and examples \ref{ex:running} and \ref{ex:running-revisit},
  see Table \ref{table:example-programs}.
\end{example}

A version of \MT-security has been studied in the past by Rafnsson and Sabelfeld
\cite{rafnsson2016secure} and Jaskelioff and Russo \cite{jaskelioff2011secure}.
These authors consider programs whose termination is stable under ``default'' inputs
in a statically labeled setting.
Specifically, they formulate definitions of security that require that the program
preserves $\ell$-equivalence and that replacing inputs with default values does not
cause non-termination.
This is analogous to \MT-security, requiring that termination is preserved 
when the high inputs have a default value, like $0$, is identical
to saying that termination is preserved when the $H$-selection of
the input is the ``default value'' $\emptyset$.
\begin{table}
  \normalsize
  \centering
  \begin{tabular}{|l|l|l|l|}
    \hline
    \textbf{Program}   & \textbf{NI} & \textbf{Termination} & \textbf{Security} \\ \hline
    $\id$              & yes         & \Total               & \Total-secure     \\ \hline
    $\combine$         & yes         & \Total               & \Total-secure     \\ \hline
    $\combineAll$      & no          & \Total               & $\times$          \\ \hline
    $\leakBit$         & no          & \Total               & $\times$          \\ \hline
    $\leakAll$         & no          & \Total               & $\times$          \\ \hline
    $\terminationLeak$ & yes         & \TI                  & \TI-secure        \\ \hline
    $\divOnLow$        & yes         & \TS                  & \TS-secure        \\ \hline
    $\divIfPre$        & yes         & \MT                  & \MT-secure        \\ \hline
    $\divIfAbs$        & yes         & \TI                  & \TI-secure        \\ \hline
  \end{tabular}
  \\
  ~
  \\
  \caption{\normalsize \label{table:example-programs} Noninterference (NI), termination, and security of example programs}
\end{table}
\eject

\subsection*{\label{sec:def-enforcement-mechanisms} \textbf{Enforcement Mechanisms}}

Next we turn our attention to mechanisms for enforcing \TX-security.
\begin{definition}[Enforcement Mechanism]
  An enforcement mechanisms $E$ is a polynomial-time total recursive function
  that takes any program $p : \Pow{A \times \L} \pto \Pow{B \times \L}$ to a
  program $E[p] : \Pow{A \times \L} \pto \Pow{B \times \L}$.
\end{definition}

Normally, we talk about two properties of enforcement mechanisms: 
security and transparency
\cite{hamlen2003computability, zanarini2013precise, SME, MF, FSME, bielova2016spot}.
Security means that applying a secure enforcement mechanism $E$ to any
program $p$ always yields a secure program $E[p]$.
Transparency, on the other hand, means that given a secure program $p$, applying
$E$ does not change the observable behaviour of $p$.
\begin{definition}[Security and Transparency]
  An enforcement mechanism $E$ is (1) \TX-secure and (2) $\TX'$-transparent
  if:
  \begin{enumerate}
    \item For all $p$, $E[p]$ is \TX-secure
    \item For all $\TX'$-secure $p$ and inputs $x$, $\sem{E[p]}(x) = \sem{p}(x)$.
  \end{enumerate}
  We say that $E$ is $\TX$-$\TX'$ if it is \TX-secure and $\TX'$-transparent.
\end{definition}
The chains of implications from propositions \ref{prop:termination-ordering}
and \ref{prop:security-ordering} are reflected in the definitions of security
and transparency for enforcement mechanisms.
\begin{proposition}
  The following two chains of implications are strict.
  $$
  \begin{array}{clccl}
             &E\ \text{is \Total-secure} & \vlinesep{40pt}{8pt} &          & E\ \text{is \TI-transparent}\\
    \implies &E\ \text{is \TS-secure}    &                      & \implies & E\ \text{is \MT-transparent}\\
    \implies &E\ \text{is \MT-secure}    &                      & \implies & E\ \text{is \TS-transparent}\\
    \implies &E\ \text{is \TI-secure}    &                      & \implies & E\ \text{is \Total-transparent}
  \end{array}
  $$
\end{proposition}
Note that the transparency ordering is the converse of the security ordering.
This is because security requires the mechanism to \emph{provide} \TX-security,
while transparency allows the mechanism to \emph{rely} on \TX-security.
%
%
Bearing these orderings in mind, if we write that $E$ is not secure, we  mean that it
is not \TI-secure, and if we write that $E$ is not transparent, we mean that it is not \Total-transparent.

\begin{example}
  We give a number of examples of secure and transparent enforcement mechanisms:
  \begin{itemize}
    \item $E[p](x) \defAs \emptyset$ is \Total-secure, but is not transparent.
    \item $E[p](x) \defAs p(x)$ is \TI-transparent, but not secure.
    \item No Sensitive Upgrade (NSU) \cite{Zdancewic:2002:PLI:935787} is
      \TI-secure but not transparent.
    \item Permissive Upgrade (PU) \cite{PU} is \TI-secure and transparent
      for more programs than NSU \cite{bielova2016taxonomy}, but still not transparent.
    \item The Hybrid Monitoring (HM) technique is \TI-secure \cite{bielova2016taxonomy}.
    \item Zanarini et al. \cite{zanarini2013precise} give an enforcement mechanism
      that is \TI-secure and \MT-transparent.
    \item SME \cite{SME} provides \MT-security and \MT-transparency (see below).
  \end{itemize}
\end{example}

\subsection*{\textbf{Multi-Execution}}

Inspired by Secure Multi-Execution (SME) \cite{SME}, we define the
enforcement mechanism $\ME$ (for ``Multi-Execution'') in our framework.
$\ME[p](x)$ runs $p$ multiple times, once for each security label $\ell$
in the lattice on appropriately censored input $x\project\ell$ that only
contains information visible to $\ell$, and uses the output of $p(x\project\ell)$
to construct the $\ell$-labeled part of final output of $\ME[p](x)$.
\begin{definition}
  $\sem{\ME[p]}(x) \defAs \bigcup\{\ \sem{p}(x\project\ell)@\ell\ |\ \ell \in \L\ \}$
\end{definition}
This mathematical definition of $\ME$ can naturally be implemented by
iterating over all labels $\ell$ in $\L$, provided $\L$ is finite.
Section \ref{sec:multi-execution} deals with the implementation in the
case where $\L$ is non-finite.

Figure \ref{fig:ME} shows a graphical rendition of $\ME[p]$ for the two-point
lattice $\{\LOW, \HIGH\}$.
The outer box represents $\ME[p](x)$, it takes both $\HIGH$ and $\LOW$ input
from the arrows on the left, and produces both $\HIGH$ and $\LOW$ output in the
arrows to the right.
Both the $\HIGH$ and $\LOW$ parts of the input are used in the top-most run of
$p$, and only the $\LOW$ part of the input is used in the bottom-most run.
The top-most run, the one given both the $\HIGH$ and $\LOW$ input, then
produces only the $\HIGH$ output, whereas the bottom-most run contributes
only the final $\LOW$ output.
Noninterference for this construction is easy to show, the $\LOW$ output can
only be influenced by the $\LOW$ input, as that run of $p$ never sees the actual
$\HIGH$ input.

For example, consider how $\ME$ works with the program
$\leakBit$ from Example \ref{ex:running} on inputs $\emptyset$ and $\{1^\HIGH\}$.
$$\leakBit(x)\ \defAs\ \text{if}\ 1^\HIGH \in x\ \text{then}\ \{1^\LOW\}\ \text{else}\ \{0^\LOW\}$$
Note that $\emptyset \sim_\LOW \{1^\HIGH\}$ and so we expect that
$\ME[\leakBit](\emptyset) \sim_\LOW \ME[\leakBit](\{1^\HIGH\})$.
As $\leakBit$ is $\Total$ we can simply compute to see that this does indeed hold.
$$
\begin{array}{ccc}
  \begin{alignedat}{2}
     &\ME[\leakBit](\emptyset)&&\project\LOW\\
  =\ &\leakBit(\emptyset\project\LOW)&&\project\LOW\\
  =\ &\leakBit(\emptyset)&&\project\LOW\\
  =\ &\{0^\LOW\}&&\project\LOW\\
  =\ &\{0^\LOW\}&&
  \end{alignedat}
  &
  \vlinesep{36pt}{40pt}
  &
  \begin{alignedat}{2}
     &\ME[\leakBit](\{1^\HIGH\})&&\project\LOW\\
  =\ &\leakBit(\{1^\HIGH\}\project\LOW)&&\project\LOW\\
  =\ &\leakBit(\emptyset)&&\project\LOW\\
  =\ &\{0^\LOW\}&&\project\LOW\\
  =\ &\{0^\LOW\}&&
  \end{alignedat}
\end{array}
$$

Furthermore, when $p$ is \TI-secure and both $\ME[p](x)$ and $p(x)$ terminate,
the output of $\ME[p](x)$ has to equal that of $p(x)$.
The $\HIGH$ part of $\ME[p](x)$ is clearly the same as the $\HIGH$ part of
$p(x)$, and because $p$ is noninterfering it follows that the $\LOW$ output of $p(x)$
doesn't depend on the $\HIGH$ input, and so the $\LOW$ output of $p(x)$ must be
the same as the $\LOW$ part of $p(x\project\LOW)$.
Putting these two facts together allows us to conclude that $\ME[p](x)$ is the
same as $p(x)$ when both terminate.

To see how $\ME$ deals with termination, consider program $\divIfAbs$ from Example
\ref{ex:running-termination}:
$$\divIfAbs(x)\ \defAs\ \text{if}\ 1^\HIGH \not\in x\ \text{then}\ \udef\ \text{else}\ \emptyset$$
It diverges if $1^\HIGH \not\in x$ and terminates with result $\emptyset$ otherwise.
$\ME[\divIfAbs]$ meanwhile runs both $\divIfAbs(x)$ and
$\divIfAbs(x\project\LOW)$, which means that $\ME[\divIfAbs]$ always diverges,
as $1^\HIGH \not\in x\project\LOW$ for all $x$.

However, if we consider $\ME[\divIfPre]$ we get a more interesting result.
Unlike $\divIfAbs$, $\divIfPre$ diverges when $1^\HIGH \in x$ and so
if $1^\HIGH \not\in x$, then we have $\defined{\ME[\divIfPre]}{x}$ with
output $\emptyset$.
As termination of $\ME[\divIfPre](x)$ \emph{increases} as we remove $\HIGH$
elements from $x$ and the output is constant, and $\ME[\divIfAbs](x)$ always
diverges, we see that $\ME$ behaves in an \MT-secure manner for both programs.

Monotonic Termination of $\ME[p]$ follows from the fact that $\ME[p](x)$
diverges whenever $p(x\project\ell)$ diverges for any $\ell$, and so if $\ME[p](x)$
is defined, then so is $\ME[p](x\project\ell)$.

\begin{theorem}
$\ME$ is \MT-secure.
\end{theorem}
\begin{proof}
  To show noninterference, pick any two $x$, $y$, and label $\ell$ such that
  $x \sim_\ell y$ and assume that $\defined{\ME[p]}{x}$
  and $\defined{\ME[p]}{y}$.
  For all $\jmath \canFlowTo \ell$ we have that $x \sim_\jmath y$
  and so: 
  \begin{alignat*}{2}
       &\sem{\ME[p]}(x)&&@\jmath\\
    =\ &\bigcup\{\ \sem{p}(x\project\ell')@\ell'\ |\ \ell' \in \L \}&&@\jmath\\
    =\ &\sem{p}(x\project\jmath)&&@\jmath\\
    =\ &\sem{p}(y\project\jmath)&&@\jmath\\
    =\ &\bigcup\{\ \sem{p}(y\project\ell')@\ell'\ |\ \ell' \in \L \}&&@\jmath\\
    =\ &\sem{\ME[p]}(y)&&@\jmath
  \end{alignat*}
  Because $x\project\ell = \bigcup\{x@\jmath\ |\ \jmath \canFlowTo \ell\}$
  we now have that $\sem{\ME[p]}(x) \sim_\ell \sem{\ME[p]}(y)$.

  For \MT-termination, if $\defined{\ME[p]}{x}$ then $\defined{p}{x\project\ell}$
  for all $\ell$ and so $\defined{\ME[p]}{x\project\ell}$ for all $\ell$.
  Consequently, $\ME$ is \MT-secure.
\end{proof}
Naturally, a corollary of this theorem is that $\ME$ is also \TI-secure.
\begin{corollary}
  $\ME$ is \TI-secure.
\end{corollary}

Next we address \MT-transparency.
Preservation of termination follows immediately from the definition
of the \MT~termination criteria.
This is a feature of the definition of \MT~termination, it is precisely the
definition we need in order to prove \MT-transparency.
While the observation that conditions such as \MT~ are natural pre-requisites for 
transparency (see e.g. \cite{jaskelioff2011secure,rafnsson2016secure}), we believe
that this framework provides a clean explanation for why $\MT$ exists.
Because $\ME$ runs a number of $p(x\project\ell)$ for different $\ell$, it is
natural to require that these runs terminate if $p(x)$ terminates in order to
preserve termination.
We note that \MT-security is more useful as a correctness criteria for $\ME$
than as a notion of security.
In other words, being \MT-transparent is more useful than being \MT-secure.

\begin{theorem}
  \label{thm:ME-MT-transparent}
  $\ME$ is \MT-transparent.
\end{theorem}
\begin{proof}
  Pick any \MT-secure program $p$ and input $x$.
  We first prove that $\defined{\ME[p]}{x} \iff \defined{p}{x}$.
  If $\notdefined{p}{x}$ then because $x = x\project\bigsqcup\L(x)$
  we have that $\notdefined{p}{x\project\bigsqcup\L(x)}$
  and so $\notdefined{\ME[p]}{x}$.
  If $\defined{p}{x}$ then $\defined{p}{x\project\ell}$ for all $\ell$, and
  so clearly $\defined{\ME[p]}{x}$.

  Next we show that when $\defined{\ME[p]}{x}$ and $\defined{p}{x}$
  we have that $\ME[p](x) = p(x)$.
  For all $a^\ell \in \sem{p}(x)$ is equivalent to $a^\ell \in \sem{p}(x\project\ell)@\ell$
  (by $p$ \MT-secure), which in turn is equivalent to $a^\ell \in \sem{\ME[p]}(x)$ (by definition).
  Consequently, $\sem{p}(x) = \sem{\ME[p]}(x)$.
\end{proof}
Again, a simple corollary follows from the hierarchy of notions of security.
\begin{corollary}
  $\ME$ is \TS- and \Total-transparent.
\end{corollary}

Note that \MT-secure and \MT-transparent is the strongest classification of
the security and transparency of the $\ME$ enforcement mechanism.
For example, $\ME$ could not be \MT-secure and \TI-transparent, as this would
require preserving \TI-termination, which would in turn prohibit $\ME$ from being \MT-secure.
Likewise, $\ME$ could not be \TS-secure, as this would prohibit it from preserving the semantics
of merely \MT-secure programs.
In general, $\ME$ could not be $\TX$-$\TX'$ for any $\TX$ stronger than, or $\TX'$ weaker,
than $\MT$.

In the following few sections, we establish a number of results about the
possibility and impossibility of secure, transparent, and efficient information
flow control.
Before doing so, we review the state of transparent enforcement
of noninterference in the literature.

\subsection*{\textbf{The State of Transparent Enforcement}}
\newcommand{\tkg}{}
\newcommand{\imp}{$\times$\tkg}
\newcommand{\impM}{\imp~(\hspace{-1sp}\cite{ngo2018impossibility})}
\newcommand{\impC}{\imp~\cite{ngo2018impossibility}}
\newcommand{\ckg}{}
\newcommand{\ok}{\ckg$\checkmark$}
\begin{table}[t]
  \begin{tabular}{cc}
    & \textbf{Transparency} \\
    \\
    \textbf{Security} &
      {
      \normalsize
      \begin{tabular}{c|c|c|c|c|}
        ~               &  \textbf{\Total} & \textbf{\TS} & \textbf{\MT} & \textbf{\TI} \\ \hline
        \textbf{\TI}    &  \ok   & \ok \cite{OptimisingFSME}  & \ok \cite{zanarini2013precise} & \ckg\METI \\ \cline{1-5}
        \textbf{\MT}    &  \ok   & \ok  & \ckg\ME \\ \cline{1-4}
        \textbf{\TS}    &  \impM & \impC \\ \cline{1-3}
        \textbf{\Total} &  \imp  \\ \cline{1-2}
      \end{tabular}
      }
    \\
  \end{tabular}
  ~
  \\
  ~
  \caption{\normalsize \label{fig:taxonomy} Possible and impossible combinations of security and transparency.}
\end{table}

Table \ref{fig:taxonomy} presents the combinations of secure and transparent
enforcement that are possible or impossible to enforce.
The table is upper-triangular as anything below the diagonal is impossible.
For example, it is impossible to simultaneously enforce a strong property such as 
\TS-security, while preserving the semantics of all programs that obey a
weaker property, such as \TI-security.
Put differently, because some programs are \TI-secure but not \TS-secure, these
programs must have their semantics altered by a \TS-secure enforcement
mechanism.

As shown above, it is possible to construct an \MT-\MT~enforcement mechanism
$\ME$.
Because possibility propagates ``up and to the left'' in our table, any
\MT-\MT~enforcement mechanism is also \TI-\MT~ and \TI-\TS, for example,
we also mark all boxes to the left and above of \MT-\MT~ as possible.

Ngo et al. \cite{ngo2018impossibility} show that \TS-\TS~enforcement is
impossible.
Their proof also works to show that the weaker \TS-\Total~condition is
impossible (we mark this by putting the citation in brackets
in Table \ref{fig:taxonomy}), and consequently
\Total-\Total~ enforcement is also impossible.

One open question in the literature is whether or not \TI-\TI~enforcement
is possible.
In Section \ref{sec:TITI} we answer this question by presenting $\METI$, a \TI-\TI~
enforcement mechanism.

One issue faced by $\ME$ is that a naive implementation fails to converge
when $\L$ is non-finite, and takes infeasibly long to execute when $\L$ is very
large.
Running $p$ once for each label in an infinite lattice is, clearly, impossible.
In Section \ref{sec:multi-execution} we construct $\MEF$,
a version of $\ME$ which lifts these restrictions and makes $\ME$ work for
non-finite lattices.

\section{\label{sec:multi-execution} Multi-Execution for Decentralised Lattices}
In this section, we instantiate the framework of
Section \ref{sec:def-enforcement-mechanisms} in the form of a
novel formulation of the $\ME$ enforcement mechanism.
Our new enforcement mechanism is inspired by MF \cite{MF},
FSME \cite{FSME}, and OGMF \cite{OGMF}.
Unlike SME, these schemes do not run one parallel version of the program for
each level in the lattice, instead multi-execution under these schemes is
determined by the behaviour of the program and the levels in the input.
These methods are ``adaptively'' multi-executing in a ``white-box'' manner.
This allows them to avoid running one copy of $p$ for each security level and
thereby also allows them to work for infinite $\L$, something which SME is
unable to do.

We develop an enforcement mechanism $\MEF$ to provide a black-box account
of these methods.
$\MEF$ is, to the best of our knowledge, the first black-box enforcement
mechanism that is \MT-secure, \MT-transparent, and works for non-finite (so
called ``decentralised'' \cite{DCLabels}) lattices.

We start by observing that the definition of $\ME$ requires running $p$
on $x\project\ell$ for each $\ell \in \L$.
Assuming $x$ is finite, then there is only a finite number of such downward
projections of $x$, even for non-finite $\L$.
Thus, regardless of the size of $\L$ we only need to run $p$ a finite number
of times.
The tricky part is to appropriately compose the result of these runs to produce
the right output.

To accomplish this, we need to introduce some additional machinery.
If $S$ is a finite subset of $\L$ we call
$$
C(S) \defAs \{\ \bigsqcup S'\ |\ S' \subseteq S\ \}
$$
the \emph{closure-set of $S$}.
Note that $C(S)$ simply encodes all the ways of combining zero or more labels
in $S$, including both the lower bound of $\bot$ and the upper bound of
$\bigsqcup S$.
For example, if $S = \L(x)$ for some $x$, then $C(S)$ is all the ways to label
data arising by combining zero or more $a^\ell \in x$.

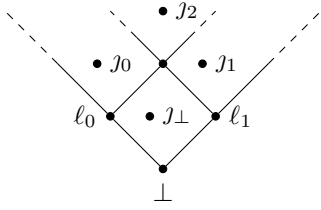
\begin{figure}[t]
  \centering
\begin{tikzpicture}[scale=0.7,dot/.style={circle,draw, fill=black, inner sep = 0pt, text width = 1mm}]
  \node[dot,label=below:{$\bot$}] at (0, 0) (bot) {};
  \node[dot,label=right:{$\jmath_\bot$}] at (-0.25, 1) (j) {};
  \node[dot,label=left:{$\ell_0$}] at (-1, 1) (l0) {};
  \node[dot,label=right:{$\jmath_0$}] at (-1.25, 2) (j0) {};
  \node[dot,label=right:{$\ell_1$}] at (1, 1)  (l1) {};
  \node[dot,label=right:{$\jmath_1$}] at (0.75, 2) (j1) {};
  \node[dot] at (0, 2)  (l2) {};
  \node[dot,label=right:{$\jmath_2$}] at (0, 3) (j2) {};
  \draw (bot) -- (l0);
  \draw (bot) -- (l1);
  \draw (l0) -- (l2);
  \draw (l1) -- (l2);
  \draw (l2) -- (-0.5,2.5);
  \draw[dashed] (-0.5,2.5) -- (-1,3);
  \draw (l2) -- (0.5,2.5);
  \draw[dashed] (0.5,2.5) -- (1,3);
  \draw (l0) -- (-2,2);
  \draw[dashed] (-2,2) -- (-3,3);
  \draw (l1) -- (2,2);
  \draw[dashed] (2,2) -- (3,3);
  \end{tikzpicture}
  \caption{\normalsize \label{fig:CL} The upwards neighbourhoods of $C(\{\ell_0, \ell_1\})$}
\end{figure}

For any set $S \subseteq \L$ and label $\ell \in S$ we define the \emph{upward neighbourhood of $\ell$ in $S$},
written $\ell \uparrow S$, as:
$$
\ell \uparrow S \defAs \{\ \jmath \in \L\
                        |\ \ell \canFlowTo \jmath,
                        \forall \jmath' \in S.\ \jmath' \canFlowTo \jmath \implies \jmath' \canFlowTo \ell\ \}
$$
That is, $\ell \uparrow S$ is the set of all $\jmath \in \L$ such that if
$\ell \canFlowTo \jmath$, then there is no other label in $S$ which
can flow to $\jmath$ but can not flow to $\ell$.
Put differently, if $\jmath \in \ell \uparrow S$, then $\ell$ is the greatest
label in $S$ such that $\ell \canFlowTo \jmath$.
From the point of view of information flow, this means that
$\ell$ captures all the labels in $s$ whose information may
flow to $\jmath$.

For an illustration of how these upwards neighbourhoods work, see Figure \ref{fig:CL}.
Each area in the diagram corresponds to the upwards neigbourhood of one of the levels
in $C(\{\ell_0, \ell_1\})$.
The level $\jmath_\bot$ in the middle square, for example, is in the upward neighbourhood
of level $\bot$ in $C(\{\ell_0, \ell_1\})$ while $\jmath_0$ is in the upward neighbourhood
of $\ell_0$.
Likewise, $\jmath_1$ is in $\ell_1 \uparrow C(\{\ell_0, \ell_1\})$
and $\jmath_2$ is in $\ell_0 \lub \ell_1 \uparrow C(\{\ell_0, \ell_1\})$.

Note that in the case where $S = \{\ell_0, \ell_1\}$ where $\ell_0$ and $\ell_1$ are
incomparable, neither $\ell_0 \uparrow S$ nor $\ell_1 \uparrow S$ contain $\ell_0 \lub \ell_1$.
For this reason, it is important to always consider upwards neighbourhoods in
$C(S)$ rather than $S$ itself.
The upward neighbourhood of $\ell$ in $C(S)$ precisely characterises the levels
for which the computation $\sem{p}(x\project\ell)$ can produce output if there
is one $\sem{p}(x\project\ell)$ for each $\ell \in C(S)$.
This is captured by the following proposition together with the fact that $\ell$
can flow to all levels in $\ell \uparrow C(S)$.
\begin{restatable}{proposition}{proppartition}
  Given any $S \subseteq \L$, the family which maps each $\ell \in C(S)$ to $\ell \uparrow C(S)$
  partitions $\L$.
\end{restatable}
\begin{proof}
  See Appendix \ref{app:proofs}.
\end{proof}

For the next definition we overload the $@$ operator to work for
sets $S \subseteq \L$ of labels as
$x@S = \{\ a^\ell\ |\ a^\ell \in x, \ell \in S\ \}$.
Finally, we define a more sophisticated version of $\ME$ that we call
$\MEF$ (for $\ME$-Fastish).
\begin{definition}
  $$\sem{\MEF}[p](x) \defAs \bigcup \{\ \sem{p}(x\project\ell)@(\ell \uparrow C(\L(x)))\ |\ \ell \in C(\L(x))\ \}$$
\end{definition}
Intuitively, this definition computes all the possible combinations of levels
which can arise from computing with data in the input (i.e. $C(\L(x))$) and
performs multi-execution with only these levels.
However, this is not sufficient to capture all the levels the data may be output at.
In effect, this is because a computation may combine data labeled $\ell_0$ and $\ell_1$,
but write the output at some level $\ell_2 \sqsupset \ell_0 \lub \ell_1$.
To get around this issue and ensure that we capture \emph{all} outputs of a computation,
we finally do the selection $\sem{p}(x\project\ell)@(\ell \uparrow C(\L(x)))$.
Because the upward neighbourhoods of $C(\L(x))$ are disjoint, the final output of $\sem{\MEF[p]}(x)$
is a union of disjoint sets, so each element of the output is added to the final set only once.
To see this principle in action, consider the following example.
\begin{example}
  Recall the definition of $\combine$ from Example \ref{ex:running}:
  $$\combine(x)\ \defAs\ \{ (|x@\Alice| + |x@\Bob|)^{\Alice \lub \Bob \lub \Charlie} \}$$
  $\MEF[\combine](\{1^\Alice\})$ runs one opy of $\combine$ for each level in $C(\{\Alice\}) = \{\bot, \Alice\}$.
  This means that $\MEF$ runs $\combine(\{1^\Alice\}\project\bot)$
  and $\combine(\{1^\Alice\}\project\Alice)$.
  The first run outputs $\{0^{\Alice \lub \Bob \lub \Charlie}\}$ and 
  the second one $\{1^{\Alice \lub \Bob \lub \Charlie}\}$.
  Because $\Alice \lub \Bob \lub \Charlie$ is in the upwards closed neighbourhood
  of $\Alice$ in $\{\Alice, \bot\}$ and not in the upwards closed neighbourhood
  of $\bot$, we have that the output of $\MEF[\combine](\{1^\Alice\})$ is
  $\{1^{\Alice \lub \Bob \lub \Charlie}\}$.
  Note that, had $\MEF$ been using the same rule for selecting outputs as
  $\ME$, i.e. had the definition used $p(x\project\ell)@\ell$ rather than
  $p(x\project\ell)@(\ell\uparrow C(\L(x)))$, the output in our example would
  be $\emptyset$, as $\Alice \lub \Bob \lub \Charlie$ is equal to neither
  $\Alice$ nor $\bot$.
\end{example}

While the formulation of $\MEF$ looks similar to that of $\ME$, it has the advantage
of being implementable even for infinite, ``decentralised'', lattices like
DC-labels \cite{DCLabels}.
It also has the nice property of being equivalent to $\ME$ (provided that both terminate).

\begin{restatable}{theorem}{thmMEFME}
  \label{thm:MEF-ME}
  For all $p, x$, if $\L$ is finite we have that:
  $$\MEF[p]{x} = \ME[p]{x}$$
\end{restatable}
\begin{proof}
  See Appendix \ref{app:proofs}.
\end{proof}

The theorem above means that $\MEF$ inherits the \MT-\MT~property of $\ME$ in the case when
$\L$ is finite.
In the case when $\L$ is non-finite, we still have that $\MEF$ is \MT-\MT.
$\MEF[p]$ is \MT-terminating because $C(\L(x\project\ell)) \subseteq C(\L(x))$ and if
$p$ is \MT-secure then termination and output behaviour are preserved as well.
The proof of \MT-transparency works by establishing that $\MEF[p](x)$ implements $\ME[p](x)$ when read
as a (non-runnable) specification for infinite lattices and so we relegate the
proof to Appendix \ref{app:proofs}.

\begin{proposition}
  \label{thm:MEF-MT-MT}
  $\MEF$ is \MT-\MT.
\end{proposition}

\section{\label{sec:TITI} Termination Insensitive Noninterference is Transparently Enforceable}
In this section we answer an open question in the literature \cite{bielova2016taxonomy}
by constructing a \TI-\TI~enforcement mechanism.
We start by briefly exploring why $\ME$ and $\MEF$ are not \TI-transparent.
The issue lies in preserving the termination of \TI-secure programs.
Put simply, if $p$ is a \TI-secure program and $p(x)$ is defined, there is no guarantee
that $p(x\project\ell)$ is defined and so there is no guarantee that $\ME[p](x)$ is defined.

Overcoming this limitation of $\ME$ boils down to two observations.
The first observation is that $\ME$ fails to be \TI-transparent only because of
termination, which boils down to the choice to use $x\project\ell$ as the input
to the run of $p$ at level $\ell$.
Because there is no guarantee that $\defined{p}{x}$ implies $\defined{p}{x\project\ell}$,
we can't produce the output from the run at $\ell$, even though $p(x)$ is defined.
If we can choose some other $x' \sim_\ell x$ such that $\defined{p}{x'}$ instead,
we may be able to get something done.
The second observation is that if $\defined{p}{x}$, then for each $\ell$ there
exists at least one $x'$ such that $x' \sim_\ell x\project\ell$ and
$\defined{p}{x'}$, namely $x$.
We can unfortunately not simply run $p(x)$ instead of $p(x\project\ell)$
however, as this may not be secure.
Instead, we need to find $x'$ such that $x' \sim_\ell x$ and $\defined{p}{x'}$
based only on $x\project\ell$, $\ell$, and $p$.

The question now becomes, how do we pick such an $x'$?
One answer is to enumerate all $x' \sim_\ell (x\project\ell)$ in a deterministic
order and dove-tail, that is run ``in parallel'', all $p(x')$ and pick the
first $x'$ such that $p(x')$ terminates.
Note that this does not mean we pick the first $x'$ \emph{in the enumeration order},
rather we pick the $x'$ such that, during dove-tailing, $p(x')$ is the first run to
terminate.
In the rest of this section we formalise this simple idea and show that by using
this one trick, we obtain a \TI-\TI~enforcement mechanism.

If $p : \Pow{A\times\L} \pto \Pow{B\times\L}$ is a program, $\ell$ a level,
$x$ an input to $p$, and $f : \mathbb{N} \to \Pow{A\times\L}$ an enumeration of $\Pow{A\times\L}$
then the following function is partial recursive:
$$
\sem{\psi^\ell_{p,x}}(i) \defAs
\begin{cases}
  &f(i)\ \text{if}\ f(i) \sim_\ell x\ \text{and}\ \defined{p}{f(i)}\\
  &\text{divergent otherwise}
\end{cases}
$$
This means that, from Corollary 5.V(a) in \cite{rogers1967theory}, as long as
there exists an $x' \sim_\ell x$ such that $\defined{p}{x'}$, then there exists a
total recursive function $\phi^\ell_{p,x}$ with the same range as $\psi^\ell_{p, x}$.
This is where the dove-tailing happens, $\phi^\ell_{p,x}$ more or less works
by interleaving the executions of $p(f(i))$ for all $i$ until one terminates.
There are two facts about $\phi$ we use when constructing our
\TI-\TI~enforcement mechanism:
\begin{enumerate}
  \item If $x \sim_\ell y$ then $\sem{\phi^\ell_{p, x\project\ell}}(i) = \sem{\phi^\ell_{p, y\project\ell}}(i)$.
  \item If $p$ is \TI-secure, then $\sem{p}(\phi^\ell_{p, x\project\ell}(i)) \sim_\ell \sem{p}(x)$ provided
    both runs terminate.
\end{enumerate}
From these two observations we define our \TI-\TI~enforcement mechanism
$\METI$ for ``multi-execution search for termination''.
\begin{definition}
  Define the enforcement mechanism $\METI$ as:
  $$
  \sem{\METI[p]}(x) \defAs
  \begin{cases}
    &\text{divergent}\ \text{if}\ \notdefined{p}{x}\\
    &\cup\{\ \sem{p}(\phi^\ell_{p,x\project\ell}(0)) @ (\ell \uparrow C(\L(x)))\\
    &\ \ |\ \ell \in C(\L(x)) \}\;\; \text{otherwise}
  \end{cases}
  $$
\end{definition}
To see how $\METI$ works, consider again the program $\divIfPre$ from Example
\ref{ex:running-termination}.
$$
\divIfPre(x) \defAs \text{if}\ 1^\HIGH \in x\ \text{then}\ \udef\ \text{else}\ \emptyset
$$
This program is \TI-secure, and so we expect that $\METI$ should preserve its semantics.
Running $\METI[\divIfPre](\{1^\HIGH\})$ diverges, as
$\notdefined{\divIfPre}{\{1^\HIGH\}}$, and assuming without loss of generality
that $\phi^\ell_{\divIfPre,\emptyset}(0) = \emptyset$ for all $\ell$, running
$\METI[\divIfPre](\emptyset)$ yields the expected result of $\emptyset$.

The proof of security for $\METI$ relies on the first observation above, and the proof
of transparency relies on the second.
\begin{theorem}
  $\METI$ is \TI-secure
\end{theorem}
\begin{proof}
  To see that $\METI$ is \TI-secure, consider any level $\ell$, program $p$, and
  inputs $x \sim_\ell y$.
  It is the case that for all $\ell' \canFlowTo \ell$ the greatest $\jmath$ in
  $C(\L(x))$ and $C(\L(y))$ such that $\jmath \canFlowTo \ell'$ are
  the same, as $\L(x\project\ell') = \L(y\project\ell')$.
  If both $\defined{\METI[p]}{x}$ and $\defined{\METI[p]}{y}$ then:
  \begin{alignat*}{2}
       &\sem{\METI[p]}(x)&&@\ell'\\
    =\ &\sem{p}(\phi^{\jmath}_{p,x\project\jmath}(0))&&@\ell'\\
    =\ &\sem{p}(\phi^{\jmath}_{p,y\project\jmath}(0))&&@\ell'\\
    =\ &\sem{\METI[p]}(y)&&@\ell'
  \end{alignat*}
  From which we conclude that $\sem{\METI[p]}(x) \sim_\ell \sem{\METI[p]}(y)$, in other words
  $\METI$ is \TI-secure.
\end{proof}
\begin{theorem}
  $\METI$ is \TI-transparent.
\end{theorem}
\begin{proof}
  Note that $\defined{p}{x} \iff \defined{\METI[p]}{x}$.
  If $p$ is \TI-secure then
  it holds that for all $\ell$ and $x \sim_\ell y$,
  $p(y)@\ell = p(x)@\ell$.
  Consider now the greatest $\jmath \in C(\L(x))$ such
  that $\jmath \canFlowTo \ell$, if $y \sim_{\jmath} x$ then
  $y \sim_\ell x$ and so $\sem{\phi^{\jmath}_{p,x\project\jmath}}(0) \sim_\ell x$
  and so, for all $\ell$:
  $$
  \METI[p](x)@\ell = p(\phi^{\jmath}_{p,x\project\jmath}(0))@\ell =p(x)@\ell
  $$
  %
  %
  Which gives us that $\METI[p](x) = p(x)$.
\end{proof}

While $\METI$ answers the question of whether or not \TI-\TI~enforcement
is possible, it does so in a somewhat unsatisfactory manner.
Computing $\METI[p](x)$ is very slow, not only because we are computing
something once for each $\ell$ in $C(\L(x))$ (which is an exponential number
of levels!), but also because the computation of $\sem{\phi^\ell_{p,x\project\ell}}(0)$
requires interleaving several (possibly an exponential number!) of runs of $p$.

\section{\label{sec:level-assignments} Transparent Enforcement is Multi-Execution}
Consider again the program $\combineAll$ from Example \ref{ex:running}.
$$\combineAll(x) \defAs \{|x|^{\bigsqcup\L(x)}\}$$
Running $\MEF[\combineAll](x)$ results in one run of $\combineAll$ for each level in
$C(\L(x))$ and thus, $\MEF[\combineAll](x)$ produces one output for each level
in $C(\L(x))$.
In the worst case, the size of the output of $\MEF[\combineAll](x)$ is exponential in the size of $x$ whereas
the original output of $p$ is polynomial sized in the size of $x$.

This observation that $\ME$ and $\MEF$ cause exponential overhead in
the output of some programs causes severe issue for multi-execution.
Many applications ``in the wild'' have large legacy code bases with vulnerabilities
and would therefore benefit from applying transparent IFC techniques.
If these techniques introduce extraneous overhead, applying them to large legacy systems
is out of the question.

As insecure programs cause $\ME$ and $\MEF$ to produce exponentially sized outputs, we
know there can be no \emph{semantics preserving} optimisation of these methods that
gets rid of this overhead.
This means that the work on OGMF \cite{OGMF} and some of the ``data-oriented''
optimisations of Algehed et al. \cite{OptimisingFSME}, can not hope to make
multi-execution practical \emph{on their own}.
Such techniques may be integral to the eventual success of multi-execution, but
not without incorporating other measures as well.

There is still hope that we may find some efficient enforcement mechanism,
but this mechanism may need to behave differently to $\ME$
on programs such as $\combineAll$ above.
Insecure programs aside, we begin by defining the least restrictive notion
of an efficient enforcement mechanism, that it ``only'' produce polynomial
overhead on already secure programs.
\begin{definition}
  An enforcement mechanism $E$ is $\TX$-efficient if and only if $E[p](x)$
  runs in time polynomial in $|p| + |x|$ whenever $p$ is \TX-secure
  and runs in time polynomial in $|x|$.
\end{definition}

In this section we consider the connection between arbitrary enforcement
mechanisms and multi-execution and show that any enforcement mechanism for
\TI-security and \MT-transparency gives rise to an efficient strategy for
multi-execution.
In spirit, it is similar to $\MEF$, we do multi-execution only for the levels
that are necessary.
We begin by generalising the definition of $\ME$ to allow it to selectively
multi-execute at particular levels.

\begin{definition}
  A partial function $\LA$ is a \emph{Level Assignment} if given a program
  $p : \Pow{A \times \L} \pto \Pow{B \times \L}$ and an input
  $x \in \Pow{A \times \L}$ we have $\LA(p, x) \subseteq \L$. 
\end{definition}

Given a level assignment $\LA$, we define the enforcement mechanism
$\ME_\LA$:
$$\sem{\ME_\LA[p]}(x) \defAs \bigcup \{\ \sem{p}(x\project\ell)@\ell\ |\ \ell \in \LA(p, x)\ \}$$
Note that $\ME_\LA$ is simply a generalisation of $\ME$.
Level assignments naturally inherit definitions from enforcement mechanisms.
\begin{definition}
  We call a level assignment $\LA$ \TX-secure if $\ME_\LA$ is \TX-secure, and \TX-transparent
  if $\ME_\LA$ is \TX-transparent.
  Furthermore, we call $\LA$ \TX-efficient if for all polynomial-time \TX-secure
  programs $\LA(p, x)$ takes time polynomial in $|p|+|x|$.
\end{definition}
Note that, if $\LA$ takes time polynomial in $|p| + |x|$ and $p$ takes time
polynomial in $|x|$, then $\ME_\LA[p](x)$ also takes time polynomial in
$|p| + |x|$.
This is because if $\LA$ is polynomial time, then it only produces a polynomially
sized output and so $\ME_\LA$ only executes a polynomial number of runs of the
polynomial time program $p$.

\begin{example}
  Deciding if a level assignment is secure and transparent can be non-trivial.
  Consider the intuitive level assignment $\LA(p,x) = \L(p(x))$.
  It allows us to compute levels without doing anything other than computing $p(x)$.
  Unfortunately, however, it is insecure in the $\HIGH-\LOW$ lattice, as illustrated
  by the following program:
  $$
  \sem{\leakLevel}(x) \defAs \text{if}\ x@\HIGH = \emptyset\ \text{then}\ \{0^\LOW\}\ \text{else}\ \emptyset
  $$
  We have that $\emptyset \sim_\LOW \{0^\HIGH\}$, but the following holds:
  $$
  \begin{array}{rcl}
    \begin{alignedat}{1}
      \leakLevel(\emptyset)        &= \{0^\LOW\}\\
      \LA(\leakLevel, \emptyset)   &= \{\LOW\}\\
      \ME_L[\leakLevel](\emptyset) &= \{0^\LOW\}
    \end{alignedat}
    & \vlinesep{20pt}{24pt} & 
    \begin{alignedat}{1}
      \leakLevel(\{0^\HIGH\})        &= \emptyset\\
      \LA(\leakLevel, \{0^\HIGH\})   &= \emptyset\\
      \ME_L[\leakLevel](\{0^\HIGH\}) &= \emptyset
    \end{alignedat}
  \end{array}
  $$
  Which gives us:
  $$
  \ME_\LA[\leakLevel](\emptyset) = \{0^\LOW\} \not\sim_\HIGH \emptyset = \ME_\LA[\leakLevel](\{0^\HIGH\})
  $$
  In conclusion, $\ME_\LA[\leakLevel]$ is not secure and so neither is $\LA$.
\end{example}

Clearly, there are some conditions which need to be met in order to produce a
secure level assignment.
The following lemma establishes one sufficient condition for $\TI$-security,
that $\LA(p,\_)$ is a secure program.
To make this precise, we observe that the notion of projection can be extended to sets
of labels by observing that $\Pow{\L}$ is isomorphic to $\Pow{1\times\L}$ where $1$
is the single-element set.
Consequently, we extend the definition of projection (and thereby also $\ell$-equivalence)
for subsets $S \subseteq \L$ as $S \project \ell = \{\ \ell' \in L\ |\ \ell' \canFlowTo \ell\ \}$.
\begin{restatable}{lemma}{lemmaLsecure}
  \label{lemma:L-secure}
  If the program $\lambda x.\ \LA(p,x)$ is \TI-secure for all $p$, then $\LA$ is \TI-secure.
\end{restatable}
\begin{proof}
  See Appendix \ref{app:proofs}. 
\end{proof}
This lemma raises two immediate questions:
\begin{enumerate}
  \item Is the implication strict?
  \item Why \TI~and not \MT-secure?
\end{enumerate}
The answer to the first question is that the implication is indeed strict, as demonstrated by
the following example.
\begin{example}
  Provided that $\L$ has at least one level $\ell$ such that $\ell \not\canFlowTo \bot$,
  there exists an \MT-secure $\LA$ for which $\LA(p,\_)$ is not secure for all $p$.
  To see this, consider the program:
  $$\sem{\nameify{empty}}(x) \defAs \emptyset$$
  and define $\LA$ to be:
  $$
  \LA(q,x) \defAs
    \begin{cases}
      \emptyset\ &\text{if}\ \ell \in \L(x)\ \text{and}\ q = \nameify{empty}\\
      \{\bot\}\  &\text{if}\ \ell \not\in \L(x)\ \text{and}\ q = \nameify{empty}\\
      \emptyset\ &\text{otherwise}
    \end{cases}
  $$
  In the case where $q = \nameify{empty}$ we have that:
  $$\sem{\ME_\LA[q]}(x) = \emptyset = \sem{\ME_\LA[q]}(y)$$
  giving us that $\ME_\LA[q]$ is noninterfering.
  Furthermore, in the case where $q \not= \nameify{empty}$, we have that $\sem{\ME_\LA[q]}(x) = \emptyset$, which
  is also \MT-secure.
\end{example}
To see the answer to the second question, consider the next example.
\begin{example}
  Consider again the two-point lattice $\LOW \canFlowTo \HIGH$.
  The level assignment $\LA(p, x) \defAs \{\ \HIGH\ |\ \HIGH \not\in \L(x)\}$ satisfies the property
  that $\lambda x.\ \LA(p, x)$ is an \MT-secure program for all $p$.
  Recall the program $\divIfAbs$:
  $$
  \divIfAbs(x) \defAs \text{if}\ 1^\HIGH \not\in x\ \text{then}\ \udef\ \text{else}\ \emptyset
  $$
  Notice that $\notdefined{\divIfAbs}{\emptyset}$ and therefore we have that:
  $$
    \begin{aligned}
      &\LA(\divIfAbs, \emptyset) = \{\HIGH\}\\
      &\notdefined{\ME_\LA[\divIfAbs]}{\emptyset}
    \end{aligned}
  $$
  Furthermore, as a general fact we have that if $\LA(p, x) = \emptyset$ then
  $\ME_\LA[p](x) = \emptyset$ for all $\LA$, $p$, and $x$.
  Therefore, the following holds:
  $$
    \begin{aligned}
      \LA(\divIfAbs, \{1^\HIGH\})           &= \emptyset\\
      \sem{\ME_\LA[\divIfAbs]}(\{1^\HIGH\}) &= \emptyset
    \end{aligned}
  $$
  However, because $\emptyset = \{1^\HIGH\}\project \LOW$, we see that
  $\ME_\LA[\divIfAbs]$ is not an \MT-secure program, as it is not
  \MT-terminating; $\ME_\LA[\divIfAbs](\{1^\HIGH\})$ is defined while
  $\ME_\LA[\divIfAbs](\{1^\HIGH\}\project\LOW)$ is not.
\end{example}

Lemma \ref{lemma:L-secure} suggests the following class of level assignment
techniques that are secure by construction.
\begin{restatable}{proposition}{propLAC}
  \label{prop:label-intersect}
  If $\LA(p, x) = C(\L(x)) \cap \LA'(p)$ for some total $\LA'$, then $\LA$ is
  an \MT-secure level assignment.
\end{restatable}
\begin{proof}
  See Appendix \ref{app:proofs}.
\end{proof}
For programs $p$ that satisfy the property that for all $x$, $\L(p(x))
\subseteq C(\L(x))$, such that $\LA'(p) \supseteq \text{domain}(\L \circ p)$ the
level assignment $\LA$ above is \MT-transparent (see Lemma
\ref{lemma:L-transparent} below).
This explains the ``computation-oriented'' optimisations of Algehed et al. \cite{OptimisingFSME}.
In their paper programs only produce outputs at levels in $C(\L(x))$ and they
attach so called ``label-expressions'', booleans expressions over lattice
elements, to program points to restrict the levels for which multi-execution is
carried out.
In one of their examples, program $p$ takes input from multiple different inputs,
but only ever writes to channels labeled with single principals like $\Alice$ or $\Bob$.
The authors then attach label-expressions on the form
$\Alice \wedge \neg \Bob \vee \neg \Alice \wedge \Bob$, which represents the set
of labels $\{\{\Alice\}, \{\Bob\}\}$ to the program $p$ during multi-execution.
Note that the sets of labels attached to a program in this scheme is constant, it does
not depend on the input $x$ and so $\LA(p, x)$ is on the form required by Proposition
\ref{prop:label-intersect}.

With an appropriate sufficient condition for security established, the natural next
question is if a similar sufficient condition exists for transparency.
\begin{restatable}{lemma}{lemmaLtransparent}
  \label{lemma:L-transparent}
  $\LA$ is \MT-transparent if for all \MT-secure programs $p$ and
  inputs $x$, we have that $\defined{p}{x}$ if and only if $\defined{\LA}{p, x}$
  and $\LA(p, x) = \L(p(x))$.
\end{restatable}
\begin{proof}
  See Appendix \ref{app:proofs}. 
\end{proof}

Next we develop one of the key contributions of this paper.
We use the theory of level assignments to show that any \TI-secure
and \MT-transparent enforcement mechanism $E$ gives rise to a \TI-secure
and \MT-transparent level assignment $\LA_E$.
Furthermore, it follows trivially that multi-executing a secure, polynomial time,
programs $p$ using $\LA_E$ is only polynomially slower than executing $E[p]$.
Concretely, what this demonstrates is that any efficient solution to the problem of
\MT-transparent enforcement gives rise to an efficient solution to the problem
of computing level assignments.
\begin{definition}
  Given an enforcement mechanism $E$, define the level assignment
  $\LA_E(p, x) \defAs \L(E[p](x))$.
\end{definition}
\begin{theorem}[\TI-Security]
  \label{thm:E-secure}
  If $E$ is \TI-secure, then $\LA_E$ is a \TI-secure level assignment.
\end{theorem}
\begin{proof}
  Immediate by Lemma \ref{lemma:L-secure}.
\end{proof}
\begin{theorem}[\MT-Transparency]
  \label{thm:E-transparent}
  If $E$ is \MT-transparent, then so is $\LA_E$.
\end{theorem}
\begin{proof}
  Immediate by Lemma \ref{lemma:L-transparent}.
\end{proof}
\begin{proposition}
  \label{prop:efficient}
  If $E[p](x)$ runs in time polynomial in $|p| + |x|$ for all polynomial
  time \MT-secure programs $p$ and inputs $x$, then $\LA_E$ is \MT-efficient.
\end{proposition}

This polynomial-time correspondence between multi-execution and
any transparent enforcement mechanism means that any limit on the efficiency 
of multi-execution is a limit on the efficiency of transparent enforcement.
This is a sobering thought, to the best of our knowledge there is no efficient
way to do multi-execution.
In fact, it is not clear that there should be any efficient
way to do it.
In a sense, the level assignment problem asks that we produce a program $\LA$
that is capable of efficiently answering arbitrarily tricky questions about program
executions.


\section{\label{sec:efficiency} Efficient Black-Box Enforcement is Impossible for Decentralised Lattices}
As seen empirically in previous work, transparent enforcement suffers
from exponential overheads \cite{MF, FSME, OptimisingFSME, wong2018faster}.
In this section we argue that there may be fundamental reasons why this is the
case by showing that it is impossible to do efficient, i.e. polynomial
overhead, black-box enforcement.
Intuitively, an enforcement mechanism $E$ is black-box if
$E[p](x)$ can only do a series of tests of $p : \Pow{A \times\L} \pto \Pow{B\times\L}$, 
where each test consists of a test-input $x_t \in \Pow{A\times\L}$ and 
a ``continuation'' which takes the result of $p(x_t)$ and either returns
a final result $E[p](x) \in \Pow{B\times\L}$ or continues testing.
Thus, we define the set of \emph{test sequences} $T$ from $A$ to $B$ inductively as:
$$
T ::= \Pow{B \times \L}\ |\ \Pow{A \times \L} \times (\Pow{B \times \L} \pto T)
$$
A test sequence $t \in T$ is either finished, $t \in \Pow{B \times \L}$, or it is
a test $a \in \Pow{A \times \L}$ together with a continuation $c \in \Pow{B \times \L} \pto T$.

For example, $\MEF$ can be understood as producing a test sequence where it tests
all $x\project\ell$ for $\ell$ in $C(\L(x))$, keeps the relevant parts of the
output around, and constructs the final union when it produces its finished result.

We define the program $\text{Exec} : \text{Program} \times T \pto \Pow{B\times\L}$,
which takes a program $p$ and a test sequence $t$ and produces the result of evaluating $t$ with $p$.
\begin{alignat*}{2}
  &\text{Exec}(p, (y, c)) &&\defAs \text{Exec}(p, c(\sem{p}(y)))\\
  &\text{Exec}(p, z)      &&\defAs z
\end{alignat*}
If $e : \Pow{A \times \L} \pto T$ takes an input and produces a test sequence,
then we can execute $e$ by running the following program $\text{Exec}(p, e(x))$:
If $\text{Exec}(p, e(x))$ terminates, we can obtain the \emph{trace} of
tested inputs and the corresponding output of $p$ (i.e. a list of $\Pow{A
\times \L} \times \Pow{B \times \L}$) by computing $\text{Trace}(p, e(x))$:
\begin{alignat*}{2}
  &\text{Trace}(p, (y, c)) &&\defAs (y, p(y)) : \text{Trace}(p, c(\sem{p}(y)))\\
  &\text{Trace}(p, z)      &&\defAs []
\end{alignat*}
Because we are interested in entirely black-box methods, we require that the
enforcement mechanism $E$ should work for any choice of lattice.
One consequence of this choice is that $E$ must \emph{discover} elements of $\L$.
We can formalise the idea that $E$ must discover its labels by defining a
trace to be consistent with an initially known set of labels $S$ inductively as:
\begin{mathpar}
  \frac{\;}{\text{Consistent}(S, [])}

  \frac{\L(y) \subseteq C(S)\;\;\;\;\text{Consistent}(S \cup \L(z), t)}{\text{Consistent}(S, (y, z) : t)}
\end{mathpar}
With these definitions in place, we can define what it means for an enforcement mechanism
to be black-box:
\begin{definition}
  $E$ is a \emph{black-box} enforcement mechanism if and only if there exists a function
  $e : \Pow{A \times \L} \pto T$ such that:
  $$
  \sem{E[p]}(x) = \text{Exec}(p, e(x))
  $$
  We call $\text{Trace}(p, e(x))$ the trace of $E$ at $p, x$
  and require that $\text{Consistent}(\emptyset, \text{Trace}(p, e(x)))$
  always holds.
\end{definition}
We refer to set of test inputs in the trace $t$ produced by $E$ at $p$ and $x$,
i.e. the set of first elements of the tuples in $t$, as the \emph{test-set} of $E$.
Note that our definition pre-supposes that $E$ initially does not know
\emph{any} of the labels in $\L$.
This restriction can be lifted to allow $E$ to have prior knowledge of a finite
proper subset of $\L$ and all the results in this section still hold.

The intuition for the proof that efficient black-box enforcement does not exist
is that all an efficient black-box enforcement mechanism $E$ can do is run the
program $p$ a polynomial number of times.
This means that there is always a large number of subsets of the input
$x$ which $E$ has not tried $p$ on.
We exploit this fact by constructing, from $E$, two secure programs $p$ and $q$
and an insecure program $w$.
We then pick two $\ell$-equivalent inputs such that $w$ behaves like $p$ on one input and
$q$ on the other.
Because $E$ is black-box, it can not tell the two runs of $w$ apart from the runs of $p$ and
$q$, and so $E$ is forced to preserve the semantics of $w$, which is insecure.

\begin{theorem}
  \label{thm:efficient-black-box-impossible}
  There is no \Total-efficient \TI-\Total~black-box enforcement mechanism.
\end{theorem}
\begin{proof}
  Let $\L = \Pow{\Nat}$, the powerset lattice over the natural numbers.
  We consider programs where the domain and co-domain are both the singleton
  set $\{1\}$.
  In other words, we consider programs $p : \Pow{\{1\} \times \L} \pto \Pow{\{1\} \times \L}$.
  
  Assume $E$ is a \Total-efficient \TI-\Total~black-box enforcement mechanism.
  For any finite subset $S$ of $\Nat$ define:
  $$
  \toInput(S) \defAs \{ 1^{\{ n \}}\ |\ n \in S \}
  $$
  Define the \Total-secure program $\sem{\nameify{empty}}(x) \defAs \emptyset$.
  Let $X(S)$ be test-set of $E$ for $\nameify{empty}, \toInput(S)$.
  Because $E$ is \Total-efficient, we can fix a sufficiently large $S$ such
  that there exists an $S' \subset S$ such that there is no $x \in X(S)$ that
  satisfies both $\toInput(S') \subseteq x$ and $\cup\L(x) = S'$, this is
  because the size of $X(S)$ is polynomial in $|S|$ but the number of choices
  for $S'$ grows exponentially in $|S|$.
  Note that $\toInput(S) \sim_{S'} \toInput(S')$, as
  $\toInput(S) \project S' = \toInput(S') = \toInput(S') \project S'$.

  Now define another \Total-secure program:
  $$
  \sem{\nameify{greater}}(x) \defAs \text{if}\ \toInput(S') \subseteq x\ \text{then}\ \{1^{S'}\}\ \text{else}\ \emptyset
  $$
  and the \TI-insecure program:
  $$
  \sem{\nameify{equal}}(x) \defAs \text{if}\ \toInput(S') \subseteq x\ \wedge\ \cup\L(x) = S'\ \text{then}\ \{1^{S'}\}\ \text{else}\ \emptyset
  $$
  Next consider two runs of $\sem{E[\nameify{equal}]}$.
  The first is $\sem{E[\nameify{equal}]}(\toInput(S))$.
  Recall that our choice of $S'$ guarantees that $E$ never tests
  $\nameify{empty}$ on any $x$ such that $\toInput(S') \subseteq x$ and
  $\cup\L(x) = S'$.
  Furthermore, $\nameify{equal}$ behaves like $\nameify{empty}$ on all other inputs.
  Hence we know that:
  $$
  E[\nameify{equal}](\toInput(S)) = E[\nameify{empty}](\toInput(S)) = \emptyset
  $$
  The second run we consider is $\sem{E[\nameify{equal}]}(\toInput(S'))$, for which $\nameify{equal}$
  behaves as $\nameify{greater}$.
  To see why, consider that $\nameify{equal}(x) \not= \nameify{greater}(x)$ implies that the
  two tests in $\nameify{equal}$ and $\nameify{greater}$ evaluate to different things
  on $x$.
  Hence, $\toInput(S') \subseteq x$ and $\cup\L(x) \not= S'$, and so $x$ must contain labels not
  in $S'$.
  However, consistency means that $E$ can not construct any such $x$ when testing
  either $\nameify{equal}$ or $\nameify{greater}$ on $\toInput(S')$.
  As a consequence, when $E$ tests $\nameify{greater}$ it never does a test
  that would reveal the difference between $\nameify{greater}$ and $\nameify{equal}$.
  This means that we can conclude that the following holds:
  $$
  \sem{E[\nameify{equal}]}(\toInput(S')) = E[\nameify{greater}](\toInput(S')) = \{ 1^{S'} \}
  $$
  In conclusion, we have that both of the following two things hold:
  $$\toInput(S) \sim_{S'} \toInput(S')$$
  $$
  E[\nameify{equal}](\toInput(S))  =
  \emptyset                        \not\sim_{S'}
  \{ 1^{S'} \}                     =
  E[\nameify{equal}](\toInput(S')) 
  $$
  Which contradictions our assumption that $E$ is \TI-secure.
\end{proof}

Because Theorem \ref{thm:efficient-black-box-impossible} proves the \emph{easiest} \TX-$\TX'$ enforcement to be impossible
to do efficiently, a simple corollary is the following.
\begin{corollary}
  Efficient black-box \TI-\TI, \TI-\MT, \TI-\TS, \MT-\MT, \MT-\TS, and \MT-\Total~enforcement
  are all impossible.
\end{corollary}

\section{\label{sec:future-work} Future Work}
There are a number of limitations to our model.
Firstly, we do not consider reactive (like \cite{SME,zanarini2013precise}) or
interactive (like \cite{rafnsson2016secure}) programs.
We expect that extending our model to cover this case will introduce new
challenges related to termination and progress channels.
For example, there is no reason to expect that reactive systems are \emph{any
easier} to secure than batch-job ones, as in the batch-job setting we have access to
everything, including the program and all the inputs, from the start.
Whereas in the reactive or interactive setting the enforcement mechanism
needs to react to new information and act on partial knowledge of the input.
Furthermore, in the interactive case, the fact that the attacker can
dynamically choose inputs to respond to the action of the enforcement mechanism may pose
additional issues for efficient transparent enforcement.

Likewise, our model does not yet deal with declassification.
Extensional accounts of declassification in multi-execution already exist
\cite{bolocsteanu2016asymmetric, rafnsson2016secure}, and we expect they can be
adapted to our setting.

The final, and we believe most important, future work is to resolve the
question of whether or not efficient transparent enforcement is at all possible
even for our simple model.
The connection between multi-execution and transparent enforcement established
in Section \ref{sec:level-assignments} means that efficiency of enforcement
and efficiency of multi-execution are intimately related.
However, to the best of our knowledge multi-execution requires something akin
to efficient program analysis, either static or dynamic, to break the
exponential blowup barrier \cite{OptimisingFSME}.
The inherent limitations posed by undecidability therefore leads us to
conjecture that efficient transparent enforcement is impossible.
\begin{conjecture}
  \Total-efficient \TI-\Total~enforcement is impossible.
\end{conjecture}
This conjecture states the strongest possible variant of this theorem in our
framework.
\TI-\Total enforcement is our weakest condition, which means that impossibility
of efficient \TI-\Total~ enforcement implies impossibility of efficient
\TX-$\TX'$ enforcement for all $\TX$ and $\TX'$ in this paper.
We believe that a proof of this conjecture could use Theorems
\ref{thm:E-secure} and \ref{thm:E-transparent} and Proposition
\ref{prop:efficient} from Section \ref{sec:level-assignments},
which we formalise in the following lemma statement.
\begin{lemma}
  Unless there is a \Total-efficient \TI-\Total~level assignment
  $\LA$ such that:
  \begin{enumerate}
    \item For all \Total-secure $p$, $\L(p(x)) = \LA(p, x)$
    \item For all $p$, $\lambda x.\ \LA(p, x)$ is a \TI-secure program
  \end{enumerate}
  There is no \Total-efficient \TI-\Total~enforcement mechanism.
\end{lemma}
The lemma uses the $\LA_E$ level-assignment from Section \ref{sec:level-assignments}.
Condition (1) follows from the fact that $\L(E[p](x)) = \L(p(x))$ when $E$ is \Total-transparent
and $p$ \Total-secure.
Condition (2) meanwhile follows from $E$ being \TI-secure.

\section{\label{sec:related-work} Related Work}
There exists a large body of work on the expressive power and precision
of different enforcement mechanisms for noninterference.
Bielova and Rezk \cite{bielova2016taxonomy} study a number of different
mechanisms and rank them in order of precision.
Ngo et al. \cite{ngo2018impossibility} show that \TS-\TS~ enforcement is
impossible.
Hamlen et al. also study the complexity classes of monitors and
enforcement mechanisms \cite{hamlen2003computability}.

Ngo et al. \cite{ngo2015runtime} provide both a generic black-box enforcement
mechanism for reactive programs that is transparent, as well as lower and
upper bounds on what hyperproperties can and cannot be black-box transparently
enforced.
While this work is impressive, it has two limitations.
Firstly, in the interest of tractability the paper partly side-steps the issue
of termination by demanding that programs always make progress.
Secondly, the paper is only concerned with what can and cannot be enforced
in a black-box manner.
We believe that a natural and exciting direction for future work is to marry
their formalism and ours to allow careful study of white-box transparency with
non-termination in a reactive setting.

A number of authors have provided multi-execution based enforcement mechanism
with subtly different \TX-$\TX'$ properties \cite{SME,MF,OGMF,FSME,
OptimisingFSME,rafnsson2016secure,jaskelioff2011secure,yang2016precise,
bielova2016spot,zanarini2013precise}.
However, only a small number of articles focus on the efficiency trade-offs related
to transparency.
Austin and Flanagan \cite{MF} compare the performance of SME and MF on a number of
micro-benchmarks.
MF was originally proposed as an optimisation of SME.
However, it has since been made clear that MF and SME have some complementary
performance behaviours \cite{FSME}.
Schmitz et al. \cite{FSME} study the time-memory trade-off in FSME, a system
which attempts to get ``the best of both worlds'' performance between MF and
SME.
Ngo et al. \cite{OGMF} show a qualitative optimisation of MF.

Pfeffer et al. \cite{pfeffer2019efficient} study a byte-code level
multi-execution technique which they optimise by selectively multi-executing
similarly to FSME \cite{FSME}.
However, they work in a statically labeled setting with small lattices, and do
are not exposed to the issues of exponential overhead we study here.

Algehed et al. \cite{OptimisingFSME} provide a framework for optimising FSME
which inspired our discussion of level-assignments.
Their ``computation-oriented'' optimisations are capable of reducing
exponential to polynomial overhead.
However, these methods have to be applied by manual analysis of the program
being optimised.
It is unlikely that the technique can be automated to work \emph{for all
programs}, as this requires a complete static analysis tool for level
assignments.

Some authors have also studied multi-execution in other contexts than security,
notably with applications to debugging \cite{wong2018faster} and program repair
\cite{wong2018beyond}.
While an implementation of these methods has been heavily optimised \cite{wong2018faster},
they suffer from the same unavoidable exponential overheads as other multi-execution techniques
discussed in this paper.

\section{\label{sec:conclusion} Conclusion}
In this paper we have presented a thorough study of the feasibility and efficiency
of transparent information flow control.
Our results indicate that while transparent IFC is possible to achieve for a
diverse set of security criteria including Monotonically Terminating
Noninterference and Termination Insensitive Noninterference, current methods
are practically infeasible for many systems due to inherent limitations on
efficiency.
We show that any enforcement mechanism that is secure and transparent also
gives rise to a secure and transparent multi-execution based enforcement with
polynomial slowdown.
We also show that any traditional multi-execution based enforcement
mechanism can expect to have large runtime overheads.
Furthermore, we prove that efficient black-box IFC is impossible.

We also pose the conjecture that transparent, efficient, and secure IFC
enforcement is impossible in general, not just in the black-box case.
We give a lemma which we think will help researchers find a proof for our
conjecture that is based on our proof that transparent enforcement is
intimately linked to multi-execution.
Our conjecture, as stated, is without reference to any hardness assumptions
for other problems in the complexity theory literature.
Readers familiar with complexity theory may note that such un-qualified lower
bounds are typically difficult to prove, and we therefore think that a proof of
a weakened version of our conjecture may be one suitable next step for the
community.

Taken together, our results indicate that the quest for transparent IFC is
unlikely to generate a panacea for securing third-party and legacy code.
This is a troubling result, it indicates that there is no ``easy way out'' for
securing existing software systems.
At the same time, we believe that our results indicate a clear future path for
the IFC and programming languages communities.
If efficient transparent enforcement is impossible, the community ought instead
to focus on tools to help software engineers build secure systems from scratch.

We also note that while our results paint a bleak picture for multi-execution
\emph{in general}, they do not mean that there are no application areas for the
techniques.
For example, when the size of the security lattice is small multi-execution can
still be applicable \cite{de2012flowfox, pfeffer2019efficient}.
The techniques are also useful in less efficiency-sensitive contexts than general
enforcement, for example in testing \cite{wong2018faster} and program
repair \cite{wong2018beyond}.
Similarly, efficient transparent enforcement being unlikely to work
does not preclude designing enforcement for the common case to achieve good
average-case performance.

\bibliographystyle{IEEEtran}
\bibliography{main.bib}

\appendix

\subsection{\label{app:proofs} Proofs}
\subsection*{Proofs showing that $C(S)$ partitions $\L$}
\begin{lemma}
  \label{lemma:disjoint:uparrow}
  Given any finite $S \subseteq \L$ and $\ell, \jmath \in S$ such that $\ell \not= \jmath$,
  $\ell \uparrow S$ is disjoint from $\jmath \uparrow S$.
\end{lemma}
\begin{proof}
  Assume $\hat\ell \in \ell \uparrow S$ and $\hat\ell \in \jmath \uparrow S$.
  Clearly $\ell \canFlowTo \hat\ell$ and $\jmath \canFlowTo \hat\ell$, and so
  $\jmath \canFlowTo \ell$ and $\ell \canFlowTo \jmath$, therefore $\ell = \jmath$,
  a contradiction.
\end{proof}
\begin{lemma}
  \label{lemma:closure:uparrow:L}
  Given any finite $S \subseteq \L$ we have that $\bigcup\{\ \ell \uparrow C(S)\ |\ \ell \in C(S)\ \} = \L$.
\end{lemma}
\begin{proof}
  Left-to-right inclusion is trivial.
  For all $\ell \in \L$ it is the case that there is a greatest element $\jmath \in C(S)$ such that
  $\jmath \canFlowTo \ell$ (as $\bot \in C(S)$) and $\ell \in \jmath \uparrow C(S)$.
\end{proof}
\proppartition*
\begin{proof}
  Immediate from lemmas \ref{lemma:disjoint:uparrow} and \ref{lemma:closure:uparrow:L}.
\end{proof}
\subsection*{Properties of $\MEF$}
\thmMEFME*
\begin{proof}
  We first prove that $\MEF[p](x) \subseteq \ME[p](x)$.
  If $a^\ell \in \sem{\MEF[p]}(x)$, then there exists some $\jmath$
  such that $a^\ell \in \sem{p}(x\project\jmath)$ and $x\project\ell = x\project\jmath$
  and so $a^\ell \in \sem{p}(x\project\ell)@\ell$ and so $a^\ell \in \sem{\ME[p]}(x)$.
  
  Next we show that $\ME[p](x) \subseteq \MEF[p](x)$.
  If $a^\ell \in \sem{\ME[p]}(x)$, then clearly $a^\ell \in \sem{p}(x\project\ell)$ and
  there is an $\jmath \in C(\L(x))$ such that $x\project\ell = x\project\jmath$
  and $\ell \in \jmath \uparrow C(\L(x))$ and, consequently, $a^\ell \in \sem{p}(x\project\jmath)@(\jmath\uparrow C(\L(x)))$
  and so $a^\ell \in \sem{\MEF[p]}(x)$.

  We have that $\MEF[p](x) \subseteq \ME[p](x)$ and $\ME[p](x) \subseteq \MEF[p](x)$ when
  both are defined, and so $\MEF[p](x) = \ME[p](x)$ provided both are defined.

  Next we show that when $\L$ is finite, $\defined{\MEF[p]}{x} \iff \defined{\ME[p]}{x}$.
  If $\defined{\MEF[p]}{x}$ then $\defined{p}{x\project\ell}$ for all $\ell \in C(\L(x))$.
  But, for all $\jmath \in \L$ it is the case that $x\project\jmath = x\project\ell$ for some
  $\ell \in C(\L(x))$ and so $\defined{p}{x\project\ell}$, and consequently $\defined{\ME[p]}{x}$.
  If $\defined{\ME[p]}{x}$ then clearly $\defined{p}{x\project\ell}$ for all
  $\ell \in C(\L(x))$ and so $\defined{\MEF[p]}{x}$.
\end{proof}
\begin{theorem}
  $\MEF$ is $\MT$-transparent.
\end{theorem}
\begin{proof}
  By Theorem \ref{thm:MEF-ME} we have that for all $p$ and $x$ it is the case
  that $\MEF[p](x) = \ME[p](x)$ if $\ME[p](x)$ and $\MEF[p](x)$ are both
  defined.
  Furthermore, by Theorem \ref{thm:ME-MT-transparent} we have that $\ME[p](x) = p(x)$ if
  $p$ is $\MT$-secure.
  If the lattice $\L$ is non-finite $\ME$ can only be read as a specification, and not an algorithm.
  However, the argument in Theorem \ref{thm:MEF-ME} suffices to show that $\MEF$ implements the
  $\ME$ specification when it terminates.
  Therefore, it remains only to show that given an $\MT$-secure $p$, it is the case that
  $\defined{\MEF[p]}{x} \iff \defined{p}{x}$.
  If $\defined{\MEF[p]}{x}$ then $\defined{p}{x\project\bigsqcup\L(x)}$ and so
  because $x = x\project\bigsqcup\L(x)$ we have that $\defined{p}{x}$.
  Conversely, if $\defined{p}{x}$ then by $p$ $\MT$-secure we know that $\defined{p}{x\project\ell}$
  for all $\ell$ and so $\defined{\MEF[p]}{x}$.
\end{proof}
\subsection*{Level Assignments}
\propLAC*
\begin{proof}
  If $x \sim_\ell y$ then $C(\L(x)) \sim_\ell C(\L(y))$ and so $\LA(p, x) \sim_\ell \LA(p, y)$.
  Consequently, $\LA(p,\_)$ is \TI-secure.
  By Lemma \ref{lemma:L-secure} we have that $\LA$ is \TI-secure and
  so $\ME_\LA$ is noninterfering.
  To see that $\ME_\LA[p]$ is \MT-terminating, observe that $\LA(p, x\project\ell) \subseteq \LA(p, x)$
  and so if $\defined{\ME_\LA[p]}{x}$ then $\defined{p}{x\project\ell}$ for $\ell \in \LA(p, x)$
  and so $\defined{\ME_\LA[p]}{x\project\ell}$.
\end{proof}
\subsection*{Lemmas for Theorems \ref{thm:E-secure} and \ref{thm:E-transparent}}
Recall the extended interpretation $\ell$-equivalence from Section \ref{sec:level-assignments} to
subsets of $\L$ based on the following definition of $S\project\ell$:
$$
S \project \ell = \{\ \ell' \in L\ |\ \ell' \canFlowTo \ell\ \}
$$
\lemmaLsecure*
\begin{proof}
  Consider any $\ell$, $p$, $x \sim_\ell y$ such that
  $\defined{\LA}{p, x}$ and $\defined{\LA}{p, y}$.
  Clearly, for all $\jmath \canFlowTo \ell$, $\jmath \in \LA(p, x)$ if and only
  if $\jmath \in \LA(p, y)$ and so if this is the case then:
  \begin{alignat*}{2}
       &\sem{\ME_\LA[p]}(x)&&@\jmath\\
    =\ &\sem{p}(x\project\jmath)&&@\jmath\\
    =\ &\sem{p}(y\project\jmath)&&@\jmath\\
    =\ &\sem{\ME_\LA[p]}(y)&&@\jmath
  \end{alignat*}
  And otherwise $\sem{\ME_\LA[p]}(x)@\jmath = \emptyset = \sem{\ME_\LA[p]}(y)@\jmath$.
  As $\ME_\LA[p](x)@\jmath = \ME_\LA[p](y)@\jmath$ for all $\jmath \canFlowTo \ell$ we
  have that $\ME_\LA[p](x) \sim_\ell \ME_\LA[p](y)$ and so $\ME_\LA[p]$ is
  noninterfering.
  Consequently, $\ME_\LA[p]$ is \TI-secure for all $p$ and then so is $\LA$.
\end{proof}
\lemmaLtransparent*
\begin{proof}
  Assume $\LA$ satisfies the stated condition and $p$ is an
  \MT-secure program.
  Clearly, $\defined{\ME_\LA[p]}{x}$ if and only if $\defined{p}{x}$.
  Furthermore $a^\ell \in \sem{p}(x)$ if and only if $a^\ell \in \sem{p}(x\project\ell)$
  which means that $a^\ell \in \sem{\ME_\LA[p]}(x)$.
  For the other direction, if $a^\ell \in \sem{\ME_\LA[p]}(x)$ then
  $a^\ell \in \sem{p}(x\project\ell)$ and so $a^\ell \in \sem{p}(x)$.
\end{proof}


\end{document}